\documentclass{birkjour}
 \newtheorem{thm}{Theorem}[section]
 \newtheorem{cor}[thm]{Corollary}
 \newtheorem{lem}[thm]{Lemma}
 \newtheorem{prop}[thm]{Proposition}
 \theoremstyle{definition}
 \newtheorem{defn}[thm]{Definition}
 \theoremstyle{remark}

 \numberwithin{equation}{section}

\usepackage[greek,british]{babel}
\usepackage[iso-8859-7]{inputenc}
\usepackage{tikz}
\usepackage{mathptmx}
\usepackage{amsmath,amscd}
\usepackage{pgfplots}
\usepgflibrary{shapes.geometric} 
\usepgflibrary[shapes.geometric] 
\usetikzlibrary{shapes.geometric} 

\newcommand{\Spin}{\mathop{\mathrm{Spin}}}
\newcommand{\Sp}{\mathop{\mathrm{Sp}}}
\newcommand{\Cl}{\mathop{\mathrm{Cl}}}

\begin{document}
\selectlanguage{british} 
%
%
%
%
%
%
%
%
%
\bibliographystyle{plain}

\title[Clifford algebras and Coxeter groups]
 {Clifford algebra  unveils a surprising geometric significance of  quaternionic root systems of Coxeter groups}

\author[Pierre-Philippe Dechant]{Pierre-Philippe Dechant}

\address{Mathematics Department \\University of York \\ Heslington, York, YO10 5GG \\ United Kingdom \\ \\ and \\ \\
Institute for Particle Physics Phenomenology \\Ogden Centre for Fundamental Physics \\ Department of Physics \\ University of Durham \\ South Road \\ Durham, DH1 3LE \\ United Kingdom}

\email{pierre-philippe.dechant@durham.ac.uk}

 

\date{\today}

\begin{abstract}
	Quaternionic representations of Coxeter (reflection) groups of ranks 3 and 4, as well as those of $E_8$, have been used extensively in the literature.
	The present paper  analyses such Coxeter groups in the Clifford Geometric Algebra framework, which affords a simple way of performing reflections and rotations whilst exposing more clearly the underlying geometry. 
	The Clifford approach shows that the quaternionic representations in fact have very simple geometric interpretations. 
	The representations of the  groups $A_1 \times A_1 \times A_1$, $A_3$, $B_3$ and $H_3$ of rank 3 in terms of pure quaternions are shown to be simply the Hodge dualised root vectors, which determine the reflection planes of the Coxeter groups.
	Two successive  reflections result in a rotation, described by the geometric product of the two reflection vectors, giving a  Clifford spinor. 
The spinors for the  rank-3 groups $A_1 \times A_1 \times A_1$, $A_3$, $B_3$ and $H_3$ yield a new simple construction of binary polyhedral groups. These in turn generate the  groups $A_1 \times A_1 \times A_1\times A_1$, $D_4$, $F_4$ and $H_4$ of rank 4, and their widely used quaternionic representations are  shown to be spinors in disguise. 
	Therefore, the Clifford geometric product in fact \emph{induces} the rank-4 groups  from the rank-3 groups.
	In particular, the groups $D_4$, $F_4$ and $H_4$  are  exceptional structures, which our study  sheds new light on.

	IPPP/12/26, DCPT/12/52\\ 
\end{abstract}

\maketitle
\tableofcontents

\section{Introduction}\label{HGA_intro}

Reflective symmetries are ubiquitous in  nature and in mathematics. They are central to the very earliest of mathematics in the guise of the Platonic and Archimedean solids, as well as to the latest endeavours of finding a fundamental physical theory, such as string theory. The current group theoretic paradigm for reflection groups is Coxeter group theory \cite{Coxeter1934discretegroups, Humphreys1990Coxeter}, which axiomatises reflections from an abstract mathematical point of view. Clifford Geometric Algebra \cite{Hestenes1966STA, LasenbyDoran2003GeometricAlgebra} is a complementary framework that instead of generality and abstraction focuses on the physical space(-time)  that we live in and its given Euclidean/Lorentzian metric. This exposes more clearly the geometric nature of many problems in mathematics and physics. In particular, Clifford Geometric Algebra has a uniquely simple formula for performing reflections. Previous research appears to have made exclusive use of one framework at the expense of the other. Here, we advocate a combined approach that results in geometric insights from Geometric Algebra that apparently have been overlooked in Coxeter theory thus far. 
Hestenes \cite{Hestenes2002PointGroups} has given a thorough treatment of point and space groups (including the Coxeter groups of rank 3 of interest here) in Geometric Algebra, and Hestenes and Holt \cite{Hestenes2002CrystGroups} have discussed the crystallographic point and space groups, from a conformal point of view. 
Our emphasis here lies on applying Geometric Algebra to the Coxeter framework, in particular the root systems and representations.  

It has been noticed that certain root systems in Coxeter group theory can be realised in terms of Hamilton's quaternions \cite{Hamilton1844, val1964homographies, Conway03onquaternions} -- most notably the $E_8$ root system -- and quaternionic representations of Coxeter groups are used extensively in the literature for problems ranging from polytopes to quasicrystals and elementary particle theory \cite{Wilson1986E8, Tits1980Quaternions,  MoodyPatera:1993b,Chen95non-crystallographicroot, Koca1989E8, Koca:1998, Koca2003A4B4F4,Koca2006H4,Koca2007Polyhedra, Koca2010Catalan, Koca2010QuasiregularII, Koca2011Chiral}. However, all that is used in this approach are the algebraic properties of the quaternion multiplication law, conjugation and inner product. It is also obvious that the quaternionic description can not extend to  arbitrary number of dimensions. 

In this paper, we show that the quaternionic rank-3 and 4 Coxeter group representations have very simple geometric interpretations as (Hodge duals of)  vectors and as spinors, respectively, which are shrowded in the purely algebraic approach followed in the literature. We demonstrate how these structures in fact arise very naturally and straightforwardly from simple geometric considerations in Geometric Algebra. Furthermore, the concepts of vectors and spinors readily generalise to arbitrary dimensions.  
There is another surprising insight that concerns the connections between the groups of rank 3 and those of rank 4. 
In  low dimensions, there are a number of exceptional structures ($E_8$; $H_4$, $F_4$, $D_4$; $H_3$, $B_3$, $A_3=D_3$), that in the conventional approach seem unrelated. Here, we demonstrate how in fact the rank-4 groups can be derived from the rank-3 groups via the geometric product of Clifford Geometric Algebra. This is complementary to the top-down approaches of projection, for instance from $E_8$ to $H_4$ \cite{Shcherbak:1988,MoodyPatera:1993b, Koca:1998, Koca:2001, DechantTwarockBoehm2011E8A4}, or of generating subgroups by deleting nodes in Coxeter-Dynkin diagrams.

This paper is organised as follows. Section \ref{HGA_Cox} introduces the fundamentals of Coxeter group theory necessary to understand the argument. 
Section \ref{HGA_quat} summarises the basic properties of quaternions, and presents the quaternionic representations used in the literature. 
In Section \ref{HGA_GA}, we introduce the basics of Geometric Algebra necessary for our discussion and set out the notation that we will use. 
In Section \ref{HGA_A13}, we begin with the simple example of the group $A_1\times A_1 \times A_1$ as an illustration of the logic behind our approach, as it exhibits the same structural features as the later cases without most of the complexity. The Geometric Algebra reflection formalism is used to generate the full reflection group from the simple roots and the geometric interpretation for the rank-3 representation in terms of pure quaternions is given. The spinors resulting from successive reflections are found to be  the Lipschitz units, which thus induce a realisation of the rank-4 group $A_1\times A_1 \times A_1\times A_1$ that corresponds precisely to the quaternionic representation.
Section \ref{HGA_A3}  follows the same procedure in the case of $A_3$, with the spinors (the Hurwitz units) inducing the group $D_4$, which is interesting from a triality point of view and central to the equivalence of the Ramond-Neveu-Schwarz and Green-Schwarz strings in superstring theory. 
In Section \ref{HGA_B3}, $B_3$ is found to give rise to the $F_4$ root system, the largest crystallographic group in four dimensions.
$H_3$ (the icosahedral group) and $H_4$ are the largest discrete symmetry groups in three and four dimensions, and are thus the most complex case. In Section \ref{HGA_H3}, we show that the $H_3$ spinors are equivalent to  the icosians, which in turn are the roots of $H_4$.
We finally consider the advantages of a Geometric Algebra framework over conventional purely algebraic considerations (Section \ref{HGA_versor}), before we conclude in Section \ref{HGA_Concl}.

\section{Coxeter Groups}\label{HGA_Cox}

\begin{defn}[Coxeter group] A {Coxeter group} is a group generated by some involutive generators $s_i, s_j \in S$ subject to relations of the form $(s_is_j)^{m_{ij}}=1$ with $m_{ij}=m_{ji}\ge 2$ for $i\ne j$. 
\end{defn}
The  finite Coxeter groups have a geometric representation where the involutions are realised as reflections at hyperplanes through the origin in a Euclidean vector space $\mathcal{E}$ (essentially  just the classical reflection groups). In particular, let $(\cdot \vert \cdot)$ denote the inner product in $\mathcal{E}$, and $\lambda$, $\alpha\in\mathcal{E}$.  
\begin{defn}[Reflections and roots] The generator $s_\alpha$ corresponds to the {reflection}
\begin{equation}\label{reflect}
s_\alpha: \lambda\rightarrow s_\alpha(\lambda)=\lambda - 2\frac{(\lambda|\alpha)}{(\alpha|\alpha)}\alpha
\end{equation}
 at a hyperplane perpendicular to the  {root vector} $\alpha$.
\end{defn}

The action of the Coxeter group is  to permute these root vectors, and its  structure is thus encoded in the collection  $\Phi\in \mathcal{E}$ of all such roots, which form a root system: 
\begin{defn}[Root system] \label{DefRootSys}
{Root systems} satisfy the  two axioms
\begin{enumerate}
\item $\Phi$ only contains a root $\alpha$ and its negative, but no other scalar multiples: $\Phi \cap \mathbb{R}\alpha=\{-\alpha, \alpha\}\,\,\,\,\forall\,\, \alpha \in \Phi$. 
\item $\Phi$ is invariant under all reflections corresponding to vectors in $\Phi$: $s_\alpha\Phi=\Phi \,\,\,\forall\,\, \alpha\in\Phi$.
\end{enumerate}
A subset $\Delta$ of $\Phi$, called {simple roots}, is sufficient to express every element of $\Phi$ via integer linear combinations with coefficients of the same sign. 
\end{defn}

$\Phi$ is therefore  completely characterised by this basis of simple roots, which in turn completely characterises the Coxeter group. The structure of the set of simple roots is encoded in the Cartan matrix, which contains the geometrically invariant information of the root system.
\begin{defn} [Cartan matrix and Coxeter-Dynkin diagram]   For a set of simple roots $\alpha_i\in\Delta$, the matrix defined by
\begin{equation}\label{CM}
	A_{ij}=2(\alpha_i| \alpha_j)/(\alpha_i| \alpha_i)
\end{equation}
is called the  {Cartan matrix}. A graphical representation of the geometric content is given by Coxeter-Dynkin diagrams, in which nodes correspond to simple roots, orthogonal ($\frac{\pi}{2}$) roots are not connected, roots at $\frac{\pi}{3}$ have a simple link, and other angles $\frac{\pi}{m}$ have a link with a label $m$ (see Fig. \ref{fig_rep} for examples). 
\end{defn}

The crystallographic Coxeter groups arise as  Weyl groups in Lie Theory. They satisfy $A_{ij}\in\mathbb{Z}$ and roots are $\mathbb{Z}$-linear combinations of simple roots. However, the root systems of the non-crystallographic Coxeter groups $H_2$, $H_3$ and $H_4$ fulfil $A_{ij}\in\mathbb{Z}[\tau]$, where $\mathbb{Z}[\tau]=\lbrace a+\tau b| a,b \in \mathbb{Z}\rbrace$ is given in terms of the  golden ratio $\tau=\frac{1}{2}(1+\sqrt{5})$. Together with its  Galois conjugate $\sigma=\frac{1}{2}(1-\sqrt{5})$, $\tau$ satisfies the quadratic equation $x^2=x+1$. In this case, all elements of $\Phi$ are  given in terms of $\mathbb{Z}[\tau]$-linear combinations of the simple roots.  These groups do not stabilise lattices in the dimension equal to their rank, but they are important in quasicrystal theory.

\section{Quaternions}\label{HGA_quat}

The Irish mathematician William Rowan Hamilton constructed the quaternion algebra $\mathbb{H}$ in his quest to generalise the advantages of complex numbers for two-dimensional geometry and analysis to three dimensions. His geometric considerations are close in spirit to Geometric Algebra, and have profoundly influenced its development. In pure mathematics, $\mathbb{H}$ is largely deemed interesting because it is one of the only four normed division algebras  $\mathbb{R}$,  $\mathbb{C}$,  $\mathbb{H}$ and  $\mathbb{O}$ (Hurwitz' theorem). 

\begin{defn}[Quaternion Basics]\label{HGA_quatdef}
Real quaternions are defined as $q=q_0+q_ie_i$, ($i=1,2,3$) where  $q_0$ and  $q_i$ are real numbers and the quaternionic imaginary units satisfy 
\begin{equation}\label{HGA_quatunits}
	e_ie_j=-\delta_{ij}+\epsilon_{ijk}e_k, (i,j,k=1,2,3).
\end{equation}
$\delta_{ij}$ and $\epsilon_{ijk}$ are the Kronecker and Levi-Civita symbols, and summation over repeated indices is implied. We will sometimes use the notation $(q_0, q_1, q_2, q_3)$ to represent a quaternion by its quadruplet of components. 
Quaternion conjugation is defined by 
\begin{equation}\label{HGA_quatconj}
	\bar{q}=q_0-q_ie_i,
\end{equation}
which equips the space of quaternions with an inner product and a norm
\begin{equation}\label{HGA_quatnorm}
(p,q)=\frac{1}{2}(\bar{p}q+p\bar{q}),	|q|^2=\bar{q}q=q_0^2+q_1^2+q_2^2+q_3^2.
\end{equation}
\end{defn}
The group of unit quaternions is therefore topologically $S^3$, and isomorphic to $SU(2)=\Spin(3)=\Sp(1)$.
It is a theorem (see, for instance, \cite{Conway03onquaternions}) that every operation in $O(3)$ can be described by one quaternion  $q$ via $x\rightarrow\bar{q}xq$ or $x\rightarrow\bar{q}\bar{x}q$, and every orthogonal transformation in $O(4)$ can be expressed by a pair of quaternions $(p,q)$ via $x\rightarrow{p}xq$ or $x\rightarrow{p}\bar{x}q$. In the literature, the  notation 
\begin{equation}
	[p,q]:x\rightarrow pxq, [p,q]^*:x\rightarrow p\bar{x}q,
\end{equation}
is conventional (see, e.g.  \cite{Koca2010QuasiregularII}), but we shall see how a Clifford algebra approach can  streamline many derivations in applications in Coxeter theory. 

The unit quaternions have a number of discrete subgroups. Group theoretically, these are the preimages of the rotational polyhedral groups (which are discrete subgroups of $SO(3)$) under the covering homomorphism from $\Spin(3)$ (which is a $\mathbb{Z}_2$-bundle over $SO(3)$). They are the binary polyhedral groups, as opposed to the full polyhedral reflection groups, which are preimages from $O(3)=SO(3)\times\mathbb{Z}_2$ (globally), and which are thus Coxeter groups.  The discrete quaternion groups therefore give concrete realisations  of the binary polyhedral groups. Here, we list the discrete groups that will be relevant later. 

\begin{defn}[Lipschitz units]\label{HGA_def_Lip}
The 8 quaternions  of the form $$(\pm 1, 0, 0 ,0) \text{ and permutations }$$ are called the Lipschitz units, and under quaternion multiplication form a realisation of the quaternion group in 8 elements (there is no polyhedron with this symmetry group). 
\end{defn}

\begin{defn}[Hurwitz units]\label{HGA_def_Hurwitz}
The 8 Lipschitz units together with the 16 unit quaternions of the form $$\frac{1}{2}(\pm 1,\pm 1,\pm 1,\pm 1)$$ are called the Hurwitz units, and realise the binary tetrahedral group of order 24. Together with the 24 `dual' quaternions of the form $$\frac{1}{\sqrt{2}}(\pm 1,\pm 1,0,0),$$ they form a group isomorphic to the binary octahedral group of order 48.
\end{defn}

\begin{defn}[Icosians]\label{HGA_def_icosian}
The 24 Hurwitz units together with the 96 unit quaternions of the form $$(0, \pm \tau, \pm 1, \pm \sigma) \text{ and even permutations,}$$ are called the Icosians. The icosian group is isomorphic to the binary icosahedral group with 120 elements. 
\end{defn}

It has been noticed that a number of root systems relevant in Coxeter group and Lie theory can be realised in terms of quaternions. Most notably, there exists an isomorphism between the 240 roots of $E_8$ and the 120 icosians and their $\tau$-multiples \cite{Wilson1986E8, Tits1980Quaternions,  MoodyPatera:1993b, Koca1989E8}. 
Koca et al. 
 have used quaternionic representations also of rank-3 and rank-4 Coxeter groups  in a large number of papers \cite{Koca:1998, Koca2003A4B4F4,Koca2006H4,Koca2007Polyhedra, Koca2010Catalan, Koca2010QuasiregularII, Koca2011Chiral}. In Fig. \ref{fig_rep} we summarise the representations of the rank-3 groups $A_1\times A_1\times A_1$, $A_3$, $B_3$ and $H_3$, as well as the rank-4 groups $D_4$, $F_4$ and $H_4$ (the extension to $A_1\times A_1\times A_1\times A_1$ is trivial). Their interests ranged from string theory to quasicrystals, and most recently to the constructions of 4-polytopes, and 3-polytopes via projection from 4-polytopes. They employ the quaternionic formalism throughout. Here, we show that their algebraically involved manipulations in fact have very straightforward geometric interpretations in terms of Geometric Algebra, which simplifies calculations, suggests a more streamlined notation and exposes the geometry at each step. 
\begin{figure}

	\includegraphics[width=12cm]{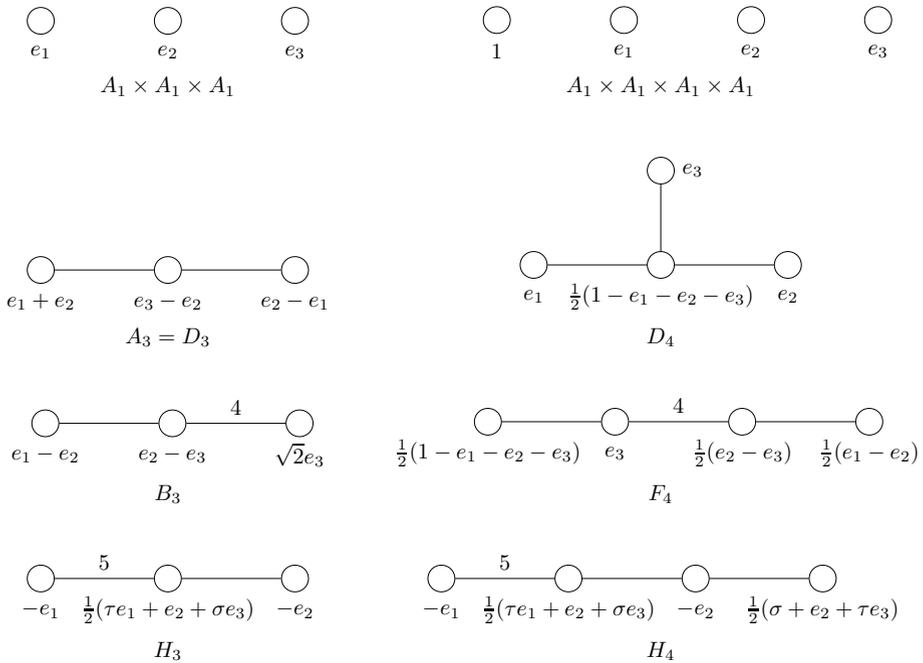}
  \caption[Hi]{Quaternionic root systems for the Coxeter groups
$A_1\times A_1\times A_1$ (as taken from  \cite{Koca2011Chiral}),
$A_3$ \cite{Koca2010QuasiregularII},
$B_3$ \cite{Koca2010QuasiregularII},
$H_3$ \cite{Koca2010QuasiregularII} of rank 3, as well as the groups
$A_1\times A_1\times A_1\times A_1$,
$D_4$ \cite{Koca2006F4},
$F_4$ \cite{Koca2006F4},
$H_4$ \cite{Koca2011Branching} of rank 4.
 }
\label{fig_rep}

\end{figure}

\section{Geometric Algebra}\label{HGA_GA}

The study of Clifford algebras and Geometric Algebra originated with Grassmann's, Hamilton's and Clifford's geometric work \cite{Grassmann1844LinealeAusdehnungslehre, Hamilton1844, Clifford1878}. However, the geometric content of the algebras was soon lost when interesting algebraic properties were discovered in mathematics, and Gibbs advocated the use of vector calculus and quaternions in physics. When Clifford algebras resurfaced in physics in the context of quantum mechanics, it was purely for their algebraic properties, and this continues in particle physics to this day. Thus, it is widely thought that Clifford algebras are somehow intrinsically quantum mechanical in nature. The original geometric meaning of Clifford algebras has been revived in the work of David Hestenes \cite{Hestenes1966STA, Hestenes1985CAtoGC, Hestenes1990NewFound}. Here, we follow an exposition along the lines of \cite{LasenbyDoran2003GeometricAlgebra}.

In a manner reminiscent of complex numbers carrying both real and imaginary parts in the same algebraic entity, one can consider the 
geometric product of two vectors defined as the sum of their scalar (inner/symmetric) product and  wedge (outer/exterior/antisymmetric) product
\begin{equation}\label{in2GP}
    ab\equiv a\cdot b + a\wedge b.
\end{equation}
The wedge product is the outer product introduced by Grassmann, as an antisymmetric product of two vectors. It is the origin of the anticommuting Grassmann `numbers', which are essential for the treatment of fermions in particle physics, from the Standard Model to supergravity and string theory. Two vectors naturally define a plane, and the wedge product is precisely a plane segment of a certain size (parallelogram or other) and orientation (a bivector). The notion of a cross-product resulting in a vector perpendicular to the plane is less fundamental, as it essentially relies on Hodge duality, and only works in three dimensions, so we concentrate here on the more general wedge product (essentially the exterior product in the language of differential forms). Unlike the constituent inner and outer products, the geometric product is invertible, as $a^{-1}$ is simply given by $a^{-1}=a/(a^2)$. This leads to many algebraic simplifications over standard vector space techniques, and also feeds through to the differential structure of the theory, with Green's function methods that are not achievable with vector calculus methods.

This geometric product can then be extended to the product of more vectors via associativity and distributivity, resulting in higher grade objects called multivectors. Multiplying a vector with itself yields a scalar, so there is a highest grade object -- called the pseudoscalar and commonly denoted by $I$ -- that one can form in the algebra, as the number of elements is limited by the number $n$ of linearly independent vectors available. The pseudoscalar is  the product of $n$ vectors, and is therefore called a grade $n$ multivector (it generalises the volume form of the exterior algebra, of which Clifford algebras are deformations). There are a total of $2^n$ elements in the algebra, which is graded into multivectors of grade $r$, depending on the number $r$ of constituent vectors.

For illustration purposes let us consider the Geometric Algebra of the plane $\Cl(2)$, generated by two orthogonal unit vectors $e_1$ and $e_2$. The outer product of two vectors $a=a_1e_1+a_2e_2$ and $b=b_1e_1+b_2e_2$ is then
\begin{equation}\label{in22DOP}
   a\wedge b = \left(a_1b_2-b_1a_2\right)e_1\wedge e_2,
\end{equation}
which would be precisely the $z$-component of the vector (cross) product in three dimensions. Due to the orthonormality we have $e_1\wedge e_2=e_1 e_2=-e_2e_1\equiv I$, i.e. orthogonal vectors anticommute. Moreover, we have the remarkable fact that $I^2=e_1e_2e_1e_2=-e_1e_2^2e_1=-1$, which is perhaps surprising when one is used to dealing with uninterpreted unit scalar imaginaries $i$. Alternatively, writing $a$ and $b$ as complex numbers $a=a_1+ia_2$ and generalising the norm $a^*a$ to the complex product $a^*b$, we have 
\begin{equation}\label{in22DCP}
  a^*b=(a_1b_1+a_2b_2)+(a_1b_2-a_2b_1)i,
\end{equation}
which just recovers the inner and outer products of the vectors, and we can in fact now identify the uninterpreted $i$ with (as?) the Geometric Algebra bivector (and pseudoscalar) $I$.

The Geometric Algebra of three dimensions $\Cl(3)$ is very similar, as there are now three planes like the one above, all  described by bivectors $e_1e_2$, $e_2e_3$ and $e_3e_1$ that square to $-1$. The new highest grade object is the trivector (and pseudoscalar) $e_1e_2e_3$, which also squares to $-1$. It might come as a surprise that the algebra obeyed by the three orthonormal vectors $e_1$, $e_2$, $e_3$ is precisely that of the Pauli algebra, with Pauli matrices $\sigma_i$. This is often thought of as something inherently quantum mechanical, but can in fact be seen as an instance of a matrix representation shrowding the geometric content of an equation. The three unit vectors $e_i$ are often therefore called $\sigma_i$, but are thought of as orthogonal unit vectors rather than matrices, which makes the algebra of three dimensions
\begin{equation}\label{in2PA}
  \underbrace{\{1\}}_{\text{1 scalar}} \,\,\ \,\,\,\underbrace{\{\sigma_1, \sigma_2, \sigma_3\}}_{\text{3 vectors}} \,\,\, \,\,\, \underbrace{\{\sigma_1\sigma_2=I\sigma_3, \sigma_2\sigma_3=I\sigma_1, \sigma_3\sigma_1=I\sigma_2\}}_{\text{3 bivectors}} \,\,\, \,\,\, \underbrace{\{I\equiv\sigma_1\sigma_2\sigma_3\}}_{\text{1 trivector}}.
\end{equation}

The geometric product provides a very compact and efficient way of handling reflections in any number of dimensions, and thus by the Cartan-Dieudonn\'e theorem also rotations. Given a unit vector $n$, we can consider the reflection of a vector $a$ in the hyperplane orthogonal to $n$.
\begin{prop}[Reflections]\label{HGA_refl}
In Geometric Algebra, a vector "$a$" transforms under a reflection in the (hyper-)plane defined by a unit normal vector "$n$" as
	\begin{equation}\label{in2refl}
	  a'=-nan.
	\end{equation}
\end{prop}

This is a remarkably compact and simple prescription for reflecting vectors in hyperplanes. More generally, higher grade multivectors transform similarly (`covariantly'), as $M= ab\dots c\rightarrow \pm nannbn\dots ncn=\pm nab\dots cn=\pm nMn$, where the $\pm$-sign defines the parity of the multivector. Even more importantly, from the  Cartan-Dieudonn\'e theorem, rotations are the product of successive reflections. For instance, compounding the reflections in the hyperplanes defined by the unit vectors $n$ and $m$ results in a rotation in the plane defined by $n\wedge m$.
\begin{prop}[Rotations]\label{HGA_rot}
In Geometric Algebra, a vector "$a$" transforms under a rotation in the plane defined by $n\wedge m$ via successive reflection in hyperplanes determined by the unit vectors "$n$" and "$m$" as 
	\begin{equation}\label{in2rot}
	  a''=mnanm=: Ra\tilde{R},
	\end{equation}
where we have defined $R=mn$ and the tilde denotes the reversal of the order of the constituent vectors $\tilde{R}=nm$.
\end{prop}

\begin{defn}[Rotors and spinors]\label{HGA_def_Rotor}
The object $R=mn$ generating the rotation in Eq. (\ref{in2rot}) is called a rotor. It satisfies $\tilde{R}R=R\tilde{R}=1$. Rotors themselves transform single-sidedly under further rotations, and thus form a multiplicative group under the geometric product, called the rotor group, which is essentially the Spin group, and thus a double-cover of the orthogonal group. Objects in Geometric Algebra that transform single-sidedly are called spinors, so that rotors are  normalised spinors. 
\end{defn}
Higher multivectors transform in the above covariant, double-sided way as $ MN\rightarrow (RM\tilde{R})(R N \tilde{R})=RM\tilde{R}R N \tilde{R}=R(MN)\tilde{R}$. In fact, the above two cases are examples of a more general theorem on the Geometric Algebra representation of orthogonal transformations. We begin with a definition:
\begin{defn}[Versor]\label{HGA_def_Versor}
In analogy to the vectors and rotors above, we define a versor as a multivector $A=a_1a_2\dots a_k$ which is the product of $k$ non-null vectors $a_i$ ($a_i^2\ne 0$). These versors also form a multiplicative group under the geometric product, called the versor group, where inverses are given by $\tilde{A}$, scaled by the magnitude $|A|^2:=|a_1|^2|a_2|^2\dots|a_k|^2=\pm A \tilde{A}$, where the sign depends on the signature of the space. The inverse is simply $A^{-1}=\pm\frac{\tilde{A}}{|A|^2}$.
\end{defn}

\begin{thm}[Versor Theorem \cite{Hestenes1990NewFound}]\label{HGA_versor_thrm}
Every orthogonal transformation $\underbar{A}$ of a vector $a$ can be expressed in the canonical form
\begin{equation}\label{in2versor}
\underbar{A}: a\rightarrow  a'=\underbar{A}(a)=\pm A^{-1}aA=\pm\frac{\tilde{A}aA}{|A|^2}.
\end{equation}
Unit versors are double-valued representations of the respective orthogonal transformation, and the $\pm$-sign defines its parity. Even versors form a double covering of the special orthogonal group, called the Spin group.
\end{thm}
The versor realisation of the orthogonal group is much simpler than conventional matrix approaches, in particular in the Conformal Geometric Algebra setup, where one uses the fact that the conformal group   $C(p,q)$ is homomorphic to $SO(p+1,q+1)$ to treat translations as well as rotations in a unified versor framework \cite{HestenesSobczyk1984,Lasenby2004Acovariantapproach}. 

The quaternions and their multiplication laws in fact arise naturally as spinors in the Clifford algebra of three dimensions.

\begin{prop}[Quaternions as spinors of $\Cl(3)$]\label{HGA_quatBV}
The unit spinors $\lbrace 1;I\sigma_1; I\sigma_2; I\sigma_3\rbrace$ of $\Cl(3)$ are isomorphic to the quaternion algebra $\mathbb{H}$. 
\end{prop}


This correspondence between quaternions and spinors will be crucial for our understanding of rank-4 root systems. We will sometimes use the notation $(a;b,c,d)$ to denote the even-grade components in $\Cl(3)$ in Eq. (\ref{in2PA}), i.e. $(a;b,c,d)=a+bI\sigma_1+cI\sigma_2+dI\sigma_3$, as opposed to the earlier -- essentially equivalent -- notation of $(q_0, q_1, q_2, q_3)$ in $\mathbb{H}$. For the interpretation of the rank-3 root systems, we make the following definition. 

\begin{defn}[Hodge dual]
In analogy with the exterior algebra, we denote multiplication of a multivector with the pseudoscalar $I$ as the \emph{Hodge dual} of the multivector. 
\end{defn}

\begin{prop}[Vectors and pure quaternions]\label{purequat}
	In three dimensions, the Hodge dual of a vector is a pure bivector. A pure bivector corresponds under the  isomorphism in Proposition \ref{HGA_quatBV} to a pure quaternion. 
\end{prop}

Having given the necessary background information, we now consider a simple geometric setup and two steps involving straightforward calculations in the Geometric Algebra of space.

\section{A simple example: the case of $A_1\times A_1 \times A_1$}\label{HGA_A13}

The procedure for this and the following sections is simply  to start with three vectors (the simple roots), then consider which other vectors are generated via reflection (the whole root system), and then to compute which rotors are generated by all the reflections.

\begin{prop}[$A_1\times A_1 \times A_1$ and the octahedron]
Take the three 3D \emph{vectors} given by $\alpha_1=(1,0,0)$, $\alpha_2=(0,1,0)$ and $\alpha_3=(0,0,1)$. The reflections in the hyperplanes orthogonal to these vectors via Eq. (\ref{in2refl}) generate further vectors pointing to the 6 vertices of an octahedron. This polyhedron is the root system of the Coxeter group $A_1\times A_1 \times A_1$, and a geometric realisation of this Coxeter group is given by the Geometric Algebra reflections in the vertex vectors.
\end{prop}

\begin{proof}
The reflection of $\alpha_1$ in $\alpha_2$ is given via \ref{HGA_refl} as $\alpha_1'=-\alpha_2\alpha_1\alpha_2=-\sigma_2\sigma_1\sigma_2=\sigma_2^2\sigma_1=\sigma_1=\alpha_1$. This does not yield a new root, but the reflection of $\alpha_1$ in itself gives $\alpha_1'=-\alpha_1\alpha_1\alpha_1=-\sigma_1^3=-\sigma_1=-\alpha_1$, and analogously for the remaining cases by explicit calculation. The vectors thus generated are $(\pm 1, 0, 0)$ and permutations thereof, which point to the vertices of an octahedron.
\end{proof}

\begin{thm}[Vertex vectors and pure quaternions]
The simple roots of the representation of $A_1\times A_1 \times A_1$ in terms of pure quaternions as shown in Fig. \ref{fig_rep} are simply the Hodge duals of the  vectors $\alpha_1$, $\alpha_2$ and $\alpha_3$.
\end{thm}

\begin{proof}
The Hodge duals of the three vectors are the bivectors $I\alpha_1$, $I\alpha_2$ and $I\alpha_3$. These pure bivectors can be mapped to the quaternions as given by Proposition \ref{HGA_quatBV}, which gives  the quaternion representation used in the literature. 
\end{proof}

Having found the full set of reflections and the geometric interpretation of the rank-3 representations in terms of pure quaternions, we now consider the rotations generated in the Geometric Algebra.

\begin{lem}[$A_1^3$ generates the Lipschitz units]
By the Cartan-Dieudonn\'e theorem, combining two reflections yields a rotation, and Proposition \ref{HGA_rot} gives a rotor realisation of these rotations in Geometric Algebra.
The 6  $A_1\times A_1 \times A_1$ Coxeter reflections generate 8 rotors, and 
since  the Spin group $\Spin(3)$ is the universal double cover of $SO(3)$, 
these 8 spinors themselves form a spinor realisation of the quaternion group $Q$ with 8 elements.
The quaternion group constitutes the  roots of the Coxeter group $A_1\times A_1 \times A_1\times A_1$.
Using the quaternion-spinor correspondence Proposition \ref{HGA_quatBV}, these 8 spinors correspond precisely to the 8 Lipschitz units in Definition \ref{HGA_def_Lip}. 
When viewed as spinor generators, the `simple roots' of the  quaternionic representation of $A_1\times A_1 \times A_1\times A_1$ from Fig. \ref{fig_rep} are seen to be sufficient to generate the quaternion group. 
\end{lem}

\begin{proof}
Form rotors according to $R_{ij}=\alpha_i\alpha_j$, e.g. $R_{11}=\alpha_1^2=1\equiv(1;0,0,0)$, or $R_{23}=\alpha_2\alpha_3=\sigma_2\sigma_3=I\sigma_1\equiv (0;1,0,0)$. Explicit calculation of all cases generates the 8 permutations of $(\pm 1;0,0,0)$. Under the spinor-quaternion correspondence, these are precisely the Lipschitz units. 
The quaternionic generators from Fig. \ref{fig_rep}, when viewed as rotor generators $R_0=1, R_1=I\sigma_1, R_2=I\sigma_2, R_3=I\sigma_3$ generate the group $A_1\times A_1 \times A_1\times A_1$, since, for instance, $R_3R_1=\sigma_1\sigma_2\sigma_2\sigma_3=-R_2$, and analogously for the other negatives. 
\end{proof}

\begin{thm}[$A_1^3$ induces $A_1^4$]
Via the geometric product, the reflections in the Coxeter group $A_1\times A_1 \times A_1$ induce spinors that realise the root system of the Coxeter group $A_1\times A_1 \times A_1\times A_1$.
In particular, the quaternionic representation in Fig. \ref{fig_rep} has a geometric interpretation as spinors mapped to the quaternions using the `accidental' isomorphism $\Cl^0(3)\sim\mathbb{H}$.
\end{thm}

\begin{proof}
Immediate. 
\end{proof}

The group $A_1\times A_1 \times A_1$ is the symmetry group of the octahedron, which is the three-dimensional orthoplex, and more generally $A_1^n$ is the symmetry group of the $n$-orthoplex, i.e. for $A_1\times A_1 \times A_1\times A_1$ the 16-cell.

\section{The case of $A_3$}\label{HGA_A3}

Having outlined the approach and the key results, we now turn to mathematically more interesting structures, beginning with $A_3$. $A_3$ is the symmetry group of the tetrahedron (the 3-simplex), and more generally, the series $A_n$ (the symmetric group) is the symmetry group of the $n$-simplex.

\begin{prop}[Cuboctahedron and $A_3$]
Take the three 3D unit {vectors} given by $\sqrt{2}\alpha_1=(1,1,0)$, $\sqrt{2}\alpha_2=(0,-1,1)$ and $\sqrt{2}\alpha_3=(-1,1,0)$. The corresponding reflections  generate further vectors pointing to the 12 vertices of a cuboctahedron. This polyhedron is the root system of the Coxeter group $A_3$, and the GA reflections in the vertex vectors give a geometric realisation.
\end{prop}

\begin{proof}
For instance, the reflection of $\alpha_1$ in itself gives $\alpha_1'=-\alpha_1\alpha_1\alpha_1=-\alpha_1$. The reflection of $\alpha_1$ in $\alpha_2$ gives  $\alpha_1''=-\alpha_2\alpha_1\alpha_2=-\frac{1}{\sqrt{8}}(\sigma_3-\sigma_2)(\sigma_2+\sigma_1)(\sigma_3-\sigma_2)=-\frac{1}{\sqrt{8}}(\sigma_3\sigma_2\sigma_3-\sigma_3\sigma_2\sigma_2+\sigma_3\sigma_1\sigma_3-\sigma_3\sigma_1\sigma_2-\sigma_2\sigma_2\sigma_3+\sigma_2\sigma_2\sigma_2-\sigma_2\sigma_1\sigma_3+\sigma_2\sigma_1\sigma_2)=\frac{1}{\sqrt{8}}(-\sigma_2-\sigma_3-\sigma_1-\sigma_3+\sigma_2-\sigma_1)=\frac{1}{\sqrt{2}}(-1,0,-1)$, and analogously for the other cases by explicit computation. The vertex vectors thus generated by the reflections are the set of permutations of $\frac{1}{\sqrt{2}}(\pm 1,\pm 1,0)$, which are the 12 vertices of a cuboctahedron.
\end{proof}

\begin{thm}[Vertices and pure quaternions]
The simple roots of the representation of $A_3$ in terms of pure quaternions as shown in Fig. \ref{fig_rep} are simply the Hodge duals of the vectors $\alpha_1$, $\alpha_2$ and $\alpha_3$.
\end{thm}

\begin{proof}
As before.
\end{proof}

\begin{lem}[$A_3$ generates the Hurwitz units]
The 12  $A_3$ Coxeter reflections generate 24 rotors, and 
these 24 spinors themselves form a spinor realisation of the binary tetrahedral group $2T$ with 24 elements.
The binary tetrahedral group constitutes the  roots of the Coxeter group $D_4$.
Under the quaternion-spinor correspondence Proposition \ref{HGA_quatBV}, these 24 spinors correspond precisely to the 24 Hurwitz units in Definition \ref{HGA_def_Hurwitz}. 
When viewed as spinor generators, the `simple roots' of the  quaternionic representation of $D_4$ from Fig. \ref{fig_rep} are seen to be sufficient to generate the binary tetrahedral group. 
\end{lem}

\begin{proof}
$R_{12}=\alpha_1\alpha_2=\frac{1}{2}(\sigma_1+\sigma_2)(\sigma_3-\sigma_2)=\frac{1}{2}(-I\sigma_2+I\sigma_1-I\sigma_3-1)=\frac{1}{2}(-1;1,-1,-1)$, and in total the  reflections generate the 16 permutations of $\frac{1}{2}(\pm 1;\pm 1,\pm 1,\pm 1)$. Together with the 8 Lipschitz units generated by $R_{13}=\alpha_1\alpha_3=\frac{1}{2}(\sigma_1+\sigma_2)(\sigma_2-\sigma_1)=I\sigma_3=(0;0,0,1)$ and similar calculations, they form the 24 Hurwitz units, which realise the binary tetrahedral group. 
The quaternionic generators from Fig. \ref{fig_rep}, when viewed as rotor generators $R_0=\frac{1}{2}(1;-1,-1,-1), R_1=I\sigma_1, R_2=I\sigma_2, R_3=I\sigma_3$ generate the group $D_4$, since, for instance, $R_1R_2=\sigma_2\sigma_3\sigma_3\sigma_1=-R_3$, and analogously for the other negatives, as well as $R_0R_1=\frac{1}{2}(\sigma_2\sigma_3-\sigma_1\sigma_2\sigma_2\sigma_3-\sigma_2\sigma_3\sigma_2\sigma_3-\sigma_3\sigma_1\sigma_2\sigma_3)=\frac{1}{2}(1;1,1,-1)$, and similar for the other permutations.
\end{proof}

\begin{thm}[$A_3$ induces $D_4$]
Via the geometric product, the reflections in the Coxeter group $A_3$ induce spinors that realise the roots of the Coxeter group $D_4$.
In particular, the quaternions in the representation of $D_4$ in Fig. \ref{fig_rep} have a geometric interpretation as spinors.
\end{thm}

$D_4$  belongs to the $D_n$-series of Coxeter groups; however, it is an accidental low-dimensional `exceptional structure' in the sense that its Dynkin diagram has an outer automorphism group of order 3 (the symmetric group $S_3$), acting by permuting the three legs (see Fig. \ref{fig_rep}). This concept, known as triality, is conventionally thought of as rather mysterious, but it is crucial in string theory, as it allows one to relate the vector representation of $SO(8)$ to the two spinor representations, needed in the proof of equivalence of the Green-Schwarz and Ramond-Neveu-Schwarz strings. Here, the Dynkin diagram symmetry is nicely manifested  in the symmetry of the pure quaternion/bivector legs, and the root system is in fact induced from 3D considerations.

\section{The case of $B_3$}\label{HGA_B3}

The $B_n$-series of Coxeter groups are the hyperoctahedral groups, the symmetry groups of the hyperoctahedra and the hypercubes. Thus, the following case $B_3$ is the symmetry group of the octahedron and the cube. 

\begin{prop}[Cuboctahedron with octahedron and $B_3$]
Take the three 3D {vectors} given by $\alpha_1=\frac{1}{\sqrt{2}}(1,-1,0)$, $\alpha_2=\frac{1}{\sqrt{2}}(0,1,-1)$ and $\alpha_3=(0,0,1)$. The corresponding reflections  generate further vectors pointing to the 18 vertices of a cuboctahedron and an octahedron. This polyhedron is the root system of the Coxeter group $B_3$, and the GA reflections in the vertex vectors give a geometric realisation.
\end{prop}

\begin{proof}
As before, reflection of a root in itself trivially yields its negative. The reflection of $\alpha_1$ in $\alpha_2$  generates $\frac{1}{\sqrt{2}}(1,0,-1)$, and further reflections generate the other vertices of a cuboctahedron. $\alpha_3$ is a vertex of an octahedron, and reflection of $\alpha_3$ in $\alpha_2$ results in $\alpha_3'=-\frac{1}{2}(\sigma_2-\sigma_3)\sigma_3(\sigma_2-\sigma_3)= -\frac{1}{2}(-\sigma_3-\sigma_2-\sigma_2+\sigma_3)=(0,1,0)$, and analogously for the other vertices of the octahedron.
\end{proof}

\begin{thm}[Vertices and pure quaternions]
The simple roots of the representation of $B_3$ in terms of pure quaternions as shown in Fig. \ref{fig_rep} are simply the Hodge duals of the vectors $\alpha_1$, $\alpha_2$ and $\alpha_3$. 
\end{thm}

\begin{lem}[$B_3$ generates the Hurwitz units and their duals]
The 18  $B_3$ Coxeter reflections generate 48 rotors, and 
these 48 spinors themselves form a spinor realisation of the binary octahedral group $2O$ with 48 elements.
The binary octahedral group constitutes the  roots of the Coxeter group $F_4$.
Under the quaternion-spinor correspondence, these 48 spinors correspond precisely to the 24 Hurwitz units and their duals. 
When viewed as spinor generators, the `simple roots' of the  quaternionic representation of $F_4$ from Fig. \ref{fig_rep} are seen to be sufficient to generate the binary octahedral group. 
\end{lem}

\begin{proof}
By explicit straightforward Geometric Algebra computation as before. 
\end{proof}

\begin{thm}[$B_3$ induces $F_4$]
Via the geometric product, the reflections in the Coxeter group $B_3$ induce spinors that realise the root system of the Coxeter group $F_4$.
In particular, the quaternions in the representation of $F_4$ in Fig. \ref{fig_rep} have a geometric interpretation as spinors.
\end{thm}

$F_4$, along with $G_2$ and $E_6\subset E_7\subset E_8$, is one of the exceptional Coxeter/Lie groups. In particular, it is the largest crystallographic symmetry in four dimensions. Thus, it is very remarkable that this exceptional structure is straightforwardly induced from $B_3$ (which is part of the $B_n$-series)  via the geometric product.

\section{The case of $H_3$}\label{HGA_H3}
This section discusses the case of $H_3$ -- the symmetry group of the icosahedron and the dodecahedron, and the largest discrete symmetry group of space -- and how it induces $H_4$, which is the largest non-crystallographic Coxeter group,  the largest Coxeter group in four dimensions, and is closely related to $E_8$. 

\begin{prop}[Icosidodecahedron and $H_3$]
Take the three 3D {vectors} given by $\alpha_1=(-1,0,0)$, $\alpha_2=\frac{1}{2}(\tau,1,\sigma)$ and $\alpha_3=(0,0,-1)$. The corresponding reflections  generate further vectors pointing to the 30 vertices of an icosidodecahedron. This polyhedron is the root system of the Coxeter group $H_3$, and the GA reflections in the vertex vectors give a geometric realisation.
\end{prop}

\begin{proof}
Analogous straightforward GA computations show that successive application of the three reflections generate the vertices $(\pm 1,0,0)$ and $\frac{1}{2}(\pm \tau,\pm 1, \pm \sigma)$ and cyclic permutations thereof, which constitute the 30 vertices of an icosidodecahedron. 
\end{proof}

\begin{thm}[Vertices and pure quaternions]
The simple roots of the representation of $H_3$ in terms of pure quaternions as shown in Fig. \ref{fig_rep} are simply the Hodge duals of the vectors $\alpha_1$, $\alpha_2$ and $\alpha_3$.
\end{thm}

\begin{lem}[$H_3$ generates the Icosians]
The 30  $H_3$ Coxeter reflections generate 120 rotors, and 
these 120 spinors themselves form a spinor realisation of the binary icosahedral group $2I$ with 120 elements.
The binary icosahedral group constitutes the  roots of the Coxeter group $H_4$.
Under the quaternion-spinor correspondence, these 120 spinors correspond precisely to the 120 icosians in Definition \ref{HGA_def_icosian}. 
When viewed as spinor generators, the `simple roots' of the  quaternionic representation of $H_4$ from Fig. \ref{fig_rep}  generate the binary icosahedral group. 
\end{lem}

\begin{proof}
As before by straightforward, if tedious, calculation. 
\end{proof}

\begin{thm}[$H_3$ induces $H_4$]
Via the geometric product, the reflections in the Coxeter group $H_3$ induce spinors that realise the roots of the Coxeter group $H_4$.
In particular, the quaternions in the representation of $H_4$ in Fig. \ref{fig_rep} have a geometric interpretation as spinors.
\end{thm}

It is remarkable that the geometric product induces $H_4$ from the  three-dimensional group $H_3$, when mostly $E_8$ is regarded as fundamental, and $H_4$ can be derived from it via Dynkin diagram foldings \cite{Shcherbak:1988,MoodyPatera:1993b, Koca:1998, Koca:2001, DechantTwarockBoehm2011E8A4}.

\section{The versor approach}\label{HGA_versor}

The Geometric Algebra approach to Coxeter groups and their quaternionic representations has a number of advantages, including simplifying calculations, providing geometric interpretations and clarity. Most importantly, the relationship between rank-3 and rank-4 groups via the geometric product does not seem to be known.

In the literature, it is sometimes remarked upon that quaternionic reflection in a unit quaternion $x$ simplifies to $v\rightarrow v-\frac{2vx}{xx}x=-x\bar{v}x$. This is the case for which the notation $[p,q]$ and $[p,q]^*$ is meant to be a shorthand, although in practice mostly $q=-p$. However, this does not achieve Geometric Algebra's simplicity for rotations and reflections, given in Eqs. (\ref{in2refl}) and (\ref{in2rot}). In order to decode what quaternionic reflection $v'=-x\bar{v}x$ actually means in practice, where $v$ and $x$ are pure quaternions, we translate to Geometric Algebra. Pure quaternions are pure bivectors, which reverse to minus themselves, so  $v'=-xv\bar{x}$. Now $v$ and $v'$ are $IV$ and $IV'$ for some vectors $V$ and $V'$, whilst $x$ is  a rotor $R$ such that $IV'=-RIV\tilde{R}$. The pseudoscalar $I$ commutes with the rotor on the left, leaving  $V'=-RV\tilde{R}$, which for $R=Ia$ reduces to a simple reflection $V'=-aVa$ in the vector $a$. $[p,-p]$ yields the rotation $V'=RV\tilde{R}$. Thus, it is not at all straightforward to see what all the  involved quaternion algebra in  the literature means, until one realises it is in fact just mostly reflections (as we are in a Coxeter framework anyway) and rotations.

\begin{table}
\begin{centering}\begin{tabular}{|c|c|c||c||c|c|c|}
\hline
rank-3 group&roots $|\Phi|$&order $|W|$&spinors&rank-4 group&roots $|\Phi|$
\tabularnewline
\hline
\hline
$A_1\times A_1\times A_1$&$6$&$8$&$8$ ($Q$)&$A_1\times A_1 \times A_1\times A_1$&$8$
\tabularnewline
\hline
$A_3$&$12$&$24$&$24$ ($2T$)&$D_4$&$24$
\tabularnewline
\hline
$B_3$&$18$&$48$&$48$ ($2O$)&$F_4$&$48$
\tabularnewline
\hline
$H_3$&$30$&$120$&$120$  ($2I$)&$H_4$&$120$
\tabularnewline
\hline
\end{tabular}\par\end{centering}
\caption[Correspondence]{\label{tab_Corr} Correspondence between the rank-3 and rank-4 Coxeter groups. $|\Phi|$ denotes the number of roots of a Coxeter group, and $|W|$ its order. The root systems $\Phi$ are the octahedron, the cuboctahedron, a cuboctahedron with an octahedron and the icosidodecahedron, respectively. The spinors generated from the rank-3 roots via the geometric product are realisations of the binary polyhedral groups $Q$, $2T$, $2O$ and $2I$ that are the roots of the rank-4 groups. In terms of quaternions, these correspond to the Lipschitz units, the Hurwitz units, the Hurwitz units with their duals, and the icosians, respectively. }
\end{table}

It is not just the geometry of the orthogonal transformations that is difficult to disentangle in the quaternion formalism, but indirectly therefore also the geometry of the (quaternionic) roots, which generate these transformations. In the literature, there is no particular significance attached to the quaternionic representations. It is merely occasionally remarked upon that it is striking that the rank-3 groups are given in terms of pure quaternions, or that the roots of the rank-4 groups are given in terms of the binary polyhedral groups, but without any further geometric insight. Here, we have shown that the pure quaternions of the rank-3 representations are merely (Hodge) disguised vertex vectors/simple roots generating the full group of Coxeter reflections, and that these generate rotor realisations of the binary polyhedral groups via the Cartan-Dieudonn\'e theorem. These spinors, in turn, are the  roots of the rank-4 Coxeter groups which are  equivalent to the simple roots of the rank-4 representations used in the literature. 
The main result of this paper is thus  the series of theorems that the geometric product induces the relevant rank-4 groups from the rank-3 groups. We summarise the correspondences between the relevant groups in Table \ref{tab_Corr}. In particular, the rank-4 groups can be generated by spinor generators that do not contain any more information than the rank-3 groups: 

\begin{prop}[Spinor generators]
The spinor generators $R_1=\alpha_1\alpha_2$ and $R_2=\alpha_2\alpha_3$  generate the four rank-4 groups from the simple roots $\alpha_i$ of the four corresponding rank-3 groups. 
\end{prop}

Despite this, it is often remarked that, for instance,  the $H_3$ roots are the special roots within $H_4$ that are pure quaternions, and that these form a sub-root system \cite{MoodyPatera:1993b,Chen95non-crystallographicroot,Twarock:2002a, Koca:2001}. 

\begin{prop}[Pure icosians]\label{pureicos}
Pure icosians are just the Hodge duals of the  roots of $H_3$, and their product is essentially that of ordinary vectors. 
\end{prop}	
\begin{proof}	
Pure icosians are pure bivectors and therefore precisely the Hodge duals of the root vectors. Their product is $I\alpha_1I\alpha_2=I^2\alpha_1\alpha_2=-\alpha_1\alpha_2$.
\end{proof}	
The real question is therefore not why the  pure quaternions form a rank-3 sub-root system of the rank-4 group, but why the Hodge duals of the vertex vectors should be amongst the spinors generated by them. This property is shared by the $A_1\times A_1 \times A_1$ case, where the Hodge duals of the vertex vectors are identical with 6 of the Lipschitz units, as well as the $B_3$ case, where the 18 Hodge dual root vectors are amongst the 48 elements of the binary octahedral group. However, the Hodge duals of the 12 cuboctahedral vertices from $A_3$ are not among the Hurwitz units. That this is not even possible for a different choice of simple roots is obvious from the form of the Hurwitz units. It is only possible for $B_3$ because there, among the 24 dual Hurwitz units  $\frac{1}{\sqrt{2}}(\pm 1,\pm 1,0,0)$, 12 are pure quaternions. However, the group $D_4$ is still induced from the group $A_3$ via the geometric product, which therefore seems to be more fundamental. Obviously, $A_3$ is contained within $D_4$ multiple times, for instance, by deleting any one leg, but one can see that this cannot be purely quaternionic, as the central node is essentially spinorial.

\begin{thm}[Pure quaternionic sub-root systems]
	The rank-3 Coxeter groups can be represented by the pure quaternion part of the representations of the rank-4 groups they generate as spinors if and only if they contain the central inversion. 
\end{thm}	
\begin{proof}
If the central inversion $I$ is generated by the roots, the bivectors defined as the Hodge duals of the roots are  pure quaternions by Proposition \ref{purequat} and they have the same product as the root vectors by Proposition \ref{pureicos}, thus giving the required sub-root system.
Conversely, if the pure quaternionic representation of the rank-3 group $I\alpha_k$ is contained in the spinors $\alpha_i\alpha_j$ generating the rank-4 group, then $I\alpha_k=\alpha_i\alpha_j$ and thus  $I=\alpha_i\alpha_j\alpha_k$ is generated by the group.
\end{proof}
	
\begin{cor}
	The Coxeter group root systems $A_1^3$, $B_3$ and $H_3$ can be realised as the pure quaternion elements of the quaternion group, the binary octahedral group and the binary icosahedral group, respectively, but $A_3$ does not have such a representation as a subset of the binary tetrahedral group.	
\end{cor}	

\begin{proof}	The Coxeter groups $A_1^3$, $B_3$ and $H_3$ contain the central inversion, but $A_3$ does not.		\end{proof}

We finally revisit the versor approach to orthogonal transformations Eq. (\ref{in2versor}). Vectors are grade 1 versors, and rotors are grade 2 versors,  so the  reflections and rotations we have been using  are  examples of the versor theorem. However, when generating the root system from the simple roots via reflection, we have in fact been forming higher grade versors already. Thus, for instance for $H_3$, the simple roots  generate 30 reflection versors (pure vectors, the root system of the Coxeter group) performing 15 different reflections via $-\alpha_i x \alpha_i$. These in turn define the 60 rotations  of the chiral icosahedral group via 120 versors acting as $\alpha_i\alpha_j x \alpha_j \alpha_i$, including the Hodge duals of the reflections (15 rotations of order 2).  60 operations of odd parity are defined by grade 3 versors (with vector and trivector parts, i.e. the Hodge duals of the 120 spinors) acting as $-\alpha_i\alpha_j\alpha_k x \alpha_k \alpha_j \alpha_i$. However, 15 of them are just the pure reflections already encountered, leaving another 45 rotoinversions. Thus, the Coxeter group, the full icosahedral group $H_3\subset O(3)$ is expressed in accordance with the versor theorem. Alternatively, one can think of 60 rotations and 60 rotoinversions, making  $I_h=I\times\mathbb{Z}_2$ manifest. However, the rotations  operate double-sidedly  on a vector, such that the versor formalism actually provides a 2-valued representation of the rotation group $SO(3)$, since the rotors $R$ and $-R$ encode the same rotation. Since $\Spin(3)$ is  the universal 2-cover of $SO(3)$, the rotors form a realisation of the preimage of the chiral icosahedral group, i.e. the binary icosahedral group. Thus, in the versor approach, we can treat all these different groups in a unified framework, whilst maintaining a clear conceptual separation.

The versor formalism is particularly powerful in the Conformal Geometric Algebra approach \cite{HestenesSobczyk1984, LasenbyDoran2003GeometricAlgebra,Dechant2011Thesis}. The conformal group $C(p,q)$ is isomorphic to $SO(p+1, q+1)$, for which one can easily construct the Clifford algebra and find rotor implementations of the conformal group action, including rotations and translations. Thus,  translations can also be handled multiplicatively as rotors, for flat, spherical and hyperbolic space-times, simplifying considerably more traditional approaches and allowing novel geometric insight. Hestenes \cite{Hestenes2002PointGroups,Hestenes2002CrystGroups, Hitzer2010CLUCalc} has applied this framework to point and space groups, which is fruitful for the crystallographic groups, as lattice translations can be treated on the same footing as the rotations and reflections. Our approach of generating the 3D Coxeter groups of interest from three vectors shares some parallels with Hestenes' approach. It also applies to the non-crystallographic case $H_3$,  which was  considered in \cite{Hestenes2002PointGroups} but not in \cite{Hestenes2002CrystGroups}. As it is non-crystallographic, it is not immediately obvious how a conformal setup would be an improvement over $\Cl(3)$, as we are not considering translations. There might perhaps be some interesting consequences for quasilattice theory \cite{Katz:1989, Senechal:1996}, in particular for  quasicrystals induced by the cut-and-project method via projection from higher dimensions \cite{MoodyPatera:1993b, DechantTwarockBoehm2011E8A4, Indelicato2011Transitions}.   A detailed treatment will be relegated to future work.

\section{Conclusions}\label{HGA_Concl}

Coxeter groups are abstract reflection groups, and Clifford Geometric Algebra affords a uniquely simple way of performing reflections, as well as other orthogonal  transformations as versors. This makes it the natural framework  for the study of reflection groups and their root systems.  Clifford Algebra provides a simple construction of the Spin group $\Spin(n)$, and in our novel construction the chiral, full  and binary  groups can  be studied in a unified framework. 

Specifically, here we have considered the simple roots of the rank-3 Coxeter groups $A_1^3$, $A_3$, $B_3$ and $H_3$, and have generated the chiral, full and binary polyhedral groups from them through the geometric product. We believe that we have presented the simplest and geometrically clearest construction of the  binary polyhedral groups known.  In particular, this constructs the roots of the rank-4 Coxeter groups $A_1^4$, $D_4$, $F_4$ and $H_4$ directly from the rank-3 roots of $A_1^3$, $A_3$, $B_3$ and $H_3$. This powerful way of inducing higher-rank Coxeter groups does not seem to be known. It is particularly interesting, as it relates the exceptional low-dimensional Coxeter groups $H_3$, $D_4$, $F_4$ and $H_4$ to each other as well as to the  series $A_n$, $B_n$ and $D_n$ in novel ways. In particular, it is remarkable that the exceptional  dimension-four phenomena are seen to arise from three-dimensional geometric considerations alone. This spinorial view could open up novel applications in polytopes (e.g. $A_4$), string theory and triality ($D_4$), lattice theory ($F_4$) and quasicrystals ($H_4$).  We have given the quaternionic representations of Coxeter groups used in the literature a geometric interpretation as vectors and spinors in disguise, and clarified and streamlined the quaternionic formalism used in the literature. We hope that our  approach will help open up new applications as a result of simplified calculations and increased geometric clarity. In particular, the concepts used here readily generalise to arbitrary dimensions.


\subsection*{Acknowledgment}
I would like to thank the anonymous referees, Eckhard Hitzer for a thorough reading of the manuscript, and Eckhard Hitzer, Anthony Lasenby, Joan Lasenby, Reidun Twarock, Mike Hobson and C\'eline B\oe hm for helpful discussions.


\end{document}